\documentclass[a4paper,reqno,11pt]{amsart}
\usepackage[a4paper,margin=1in]{geometry}
\usepackage{graphicx}
\usepackage{amsmath,amssymb,amsthm}
\usepackage{hyperref}

\theoremstyle{plain}
\newtheorem{theorem}{Theorem}[section]
\newtheorem{lemma}[theorem]{Lemma}
\newtheorem{proposition}[theorem]{Proposition}

\theoremstyle{definition}

\numberwithin{equation}{section}

\newcommand{\R}{\mathbb{R}}
\newcommand{\abs}[2][]{#1\lvert #2 #1\rvert}

\newcommand{\FF}{{\mathcal S}}

\newcommand{\ws}{w^{(s)}}
\newcommand{\phis}{\Phi^{(s)}}
\newcommand{\phistar}{\Phi^{(s_\star)}}
\newcommand{\wstar}{w^{(s_\star)}}

\title{Nonexistence of steady waves with negative vorticity}
\author{Evgeniy Lokharu}
\address{Department of Mathematics, Linköping University, SE-581 83 Linköping, Sweden}

\begin{document}
	
\begin{abstract}
	
	We prove that no two-dimensional Stokes and solitary waves exist when the vorticity function is negative and the Bernoulli constant is greater than a certain critical value given explicitly. In particular, we obtain an upper bound $F \lesssim \sqrt{2}$ for the Froude number of solitary waves with a negative constant vorticity, sufficiently large in absolute value.
	
\end{abstract}

\maketitle

\section{Introduction} \label{s:introduction}

We consider the classical water wave problem for two-dimensional steady waves with vorticity on water of finite depth. We neglect effects of surface tension and consider a fluid of constant (unit) density. Thus, in an appropriate coordinate system moving along with the wave, stationary Euler equations are given by
\begin{subequations}\label{eqn:trav}
	\begin{align}
	\label{eqn:u}
	(u-c)u_x + vu_y & = -P_x,   \\
	\label{eqn:v}
	(u-c)v_x + vv_y & = -P_y-g, \\
	\label{eqn:incomp}
	u_x + v_y &= 0, 
	\end{align}
	which holds true in a two-dimensional fluid domain
	\[
	0 < y < \eta(x).
	\]
	Here $(u,v)$ are components of the velocity field, $y = \eta(x)$ is the surface profile, $c$ is the wave speed, $P$ is the pressure and $g$ is the gravitational constant. The corresponding boundary conditions are
	\begin{alignat}{2}
	\label{eqn:kinbot}
	v &= 0&\qquad& \text{on } y=0,\\
	\label{eqn:kintop}
	v &= (u-c)\eta_x && \text{on } y=\eta,\\
	\label{eqn:dyn}
	P &= P_{\mathrm{atm}} && \text{on } y=\eta.
	\end{alignat}
\end{subequations}
It is often assumed in the literature that the flow is irrotational, that is $v_x - u_y$ is zero everywhere in the fluid domain. Under this assumption components of the velocity field are harmonic functions, which allows to apply methods of complex analysis. Being a convenient simplification it forbids modeling of non-uniform currents, commonly occurring in nature. In the present paper we will consider rotational flows, where the vorticity function is defined by
\begin{equation} \label{vort}
\omega = v_x - u_y.
\end{equation}
Throughout the paper we assume that the flow is unidirectional, that is
\begin{equation} \label{uni}
u-c > 0 
\end{equation}
everywhere in the fluid. This forbids the presence of stagnation points an gives an advantage of using the partial hodograph transform.  

In the two-dimensional setup relation \eqref{eqn:incomp} allows to reformulate the problem in terms of a stream function $\psi$, defined implicitly by relations
\[
\psi_y = u - c, \ \ \psi_x = - v.
\]
This determines $\psi$ up to an additive constant, while relations \eqref{eqn:kinbot},\eqref{eqn:kinbot} force $\psi$ to be constant along the boundaries. Thus, by subtracting a suitable constant, we can always assume that
\[
\psi = m, \ \ y = \eta; \ \ \psi = 0, \ \ y = 0.
\]
Here $m$ is the mass flux, defined by
\[
m = \int_0^\eta (u-c) dy.
\]
In what follows we will use non-dimensional variables proposed by Keady \& Norbury \cite{KeadyNorbury78}, where lengths and velocities are scaled by $(m^2/g)^{1/3}$ and $(mg)^{1/3}$ respectively; in new units $m=1$ and $g=1$. For simplicity we keep the same notations for $\eta$ and $\psi$.

Taking the curl of Euler equations \eqref{eqn:u}-\eqref{eqn:incomp} one checks that the vorticity function $\omega$ defined by \eqref{vort} is constant along paths tangent everywhere to the relative velocity field $(u-c,v)$; see \cite{Constantin11b} for more details. Having the same property by the definition, stream function $\psi$ is strictly monotone by \eqref{uni} on every vertical interval inside the fluid region. These observations together show that $\omega$ depends only on values of the stream function, that is
\[
\omega = \omega(\psi).
\]
This property and Bernoulli's law allow to express the pressure $P$ as
\begin{align}
\label{eqn:bernoulli}
P-P_\mathrm{atm} + \frac 12\abs{\nabla\psi}^2 + y  + \Omega(\psi) - \Omega(1) = const,
\end{align}
where 
\begin{align*}
\Omega(\psi) = \int_0^\psi \omega(p)\,dp
\end{align*}
is a primitive of the vorticity function $\omega(\psi)$. Thus, we can eliminate the pressure from equations and obtain the following problem:
\begin{subequations}\label{eqn:stream}
	\begin{alignat}{2}
	\label{eqn:stream:semilinear}
	\Delta\psi+\omega(\psi)&=0 &\qquad& \text{for } 0 < y < \eta,\\
	\label{eqn:stream:dyn}
	\tfrac 12\abs{\nabla\psi}^2 +  y  &= r &\quad& \text{on }y=\eta,\\
	\label{eqn:stream:kintop} 
	\psi  &= 1 &\quad& \text{on }y=\eta,\\
	\label{eqn:stream:kinbot} 
	\psi  &= 0 &\quad& \text{on }y=0.
	\end{alignat}
\end{subequations}
Here $r>0$ is referred to as Bernoulli's constant. Another constant of motion known as the flow force is given by
\begin{equation} \label{flowforce}
\FF = \int_0^{\eta}(\psi_y^2 - \psi_x^2 - y + \Omega(1) - \Omega(\psi) + r )\, dy.
\end{equation}
This constant is important in several ways; for instance, it plays the role of the Hamiltonian in spatial dynamics; see \cite{Baesens1992}. The flow force constant is also involved in a classification of steady motions; see \cite{Benjamin95}.

\subsection{Stream solutions} Laminar flows or shear currents, for which the vertical component $v$ of the velocity field is zero play an important role in the theory of steady waves. Let us recall some basic facts about stream solutions $\psi = U(y)$ and $\eta = d$, describing shear currents. It is convenient to parameterize the latter solutions by the relative speed at the bottom. Thus, we put $U_y(0) = s$ and find that $U = U(y;s)$ is subject to
\begin{equation} \label{eqn:laminar}
U'' + \omega(U) = 0, \ \ \ 0 < y < d; \ \ U(0) = 0, \ \ U(d) = 1.
\end{equation}
Our assumption \eqref{uni} implies $U' > 0$ on $[0; d]$, which puts a natural constraint on $s$. Indeed, multiplying the first equation in \eqref{eqn:laminar} by $U'$ and integrating over $[0; y]$, we find
\[
U'^2 = s^2 - 2\Omega(U).
\]
This shows that the expression $s^2 - 2 \Omega(p)$ is positive for all $p \in [0; 1]$, which requires
\[
s > s_0 = \sqrt{\max_{p \in [0,1]}2\Omega(p)}.
\]
On the other hand, every $s > s_0$ gives rise to a monotonically increasing function $U(y; s)$ solving \eqref{eqn:laminar} for some unique $d = d(s)$, given explicitly by
\[
d(s) = \int_0^1 \frac{1}{\sqrt{s^2 - 2\Omega(p)}}.
\]
This formula shows that $d(s)$ monotonically decreases to zero with respect to $s$ and takes values between zero and 
\[
d_0 = \lim_{s \to s_0+} d(s).
\]
The latter limit can be finite or not. For instance, when $\omega = 0$ we find $s_0 = 0$ and $d_0 = +\infty$. On the other hand, when $\omega = -b$ for some positive constant $b \neq 0$, then $s_0 = 0$ but $d_0 < + \infty$. We note that our main theorem is concerned with the case $d_0 < + \infty$.

Every stream solution $U(y;s)$ determines the Bernoulli constant $R(s)$, which can be found from the relation \eqref{eqn:stream:kintop}. This constant can be computed explicitly as
\[
R(s) = \tfrac12 s^2 - \Omega(1) + d(s).
\]
As a function of $s$ it decreases from $R_0$ to $R_c$ when $s$ changes from $s_0$ to $s_c$ and increases to infinity for $s>s_c$. Here the critical value $s_c$ is determined by the relation
\[
\int_0^1 \frac{1}{(s^2 - 2 \Omega(p))^{3/2}} dp = 1.
\]
The constants $R_0$ and $R_c$ are of special importance for the theory. For example, it is proved in \cite{Kozlov2015} that $r > R_c$ for any steady motion other than a laminar flow. In the present paper we will consider the water wave problem \eqref{eqn:stream} for $r > R_0$, provided $R_0 < +\infty$. The latter is true, for instance, for a negative constant vorticity.

For any $r \in (R_c, R_0]$ there are exactly two solutions $s_-(r) < s_+(r)$ to the equation
\[
R(s) = r,
\]
while for $r > R_0$ one finds only one solution $s=s_+(r)$. The laminar flow corresponding to $s_-(r)$ is called subcritical and it's depth is denoted by $d_-(r) = d(s_-(r))$. The other flow, with $s = s_+(r)$ is called supercritical and it's depth is $d_+(r) = d(s_+(r))$. According to the definition, we have
\[
d_+(r) < d_-(r).
\]
The flow force constants corresponding to flows with $s = s_\pm$ are denoted by $S_\pm(r)$.

It was recently proved in \cite{KozLokhWheeler2020} that all solitary waves are supported by supercritical depths $d_+(r)$ and the corresponding flow force constant equals to $\FF_+(r)$; here $r$ is the Bernoulli constant of a solitary wave.
\subsection{Formulations of main results.}

Just as in \cite{Kozlov2015} we split the set of all vorticity functions into three classes as follows:

\begin{itemize}
	\item[(i)] $\max_{p \in [0,1]} \Omega(p)$ is attained either at an inner point of $(0,1)$ or at an end-point, where $\omega$ attains zero value;
	\item[(ii)] $\Omega(p) < 0$ for all $p \in (0,1]$ and $\omega(0) \neq 0$;
	\item[(iii)] $\Omega(p) < \Omega(1)$ for all $p \in [0,1)$ (and so $\omega(1) \neq 0$).
\end{itemize}

The first class can be characterized by relations $R_0 = +\infty$ and $d_0 = +\infty$, while $R_0, d_0 < + \infty$ for all vorticity functions that belong to the second and third classes. Our main result states 

\begin{theorem}\label{thm:main} Let $\omega \in C^{\gamma}([0,1])$ satisfies conditions (ii) or (iii). Then there exist no Stokes waves with $r \geq R_0 - \Omega(1)$. Furthermore, there are no solitary waves with $r \geq R_0$.
\end{theorem}

A part of the statement, when $\omega$ is subject to (iii) was proved in \cite{Kozlov2015}, where it was shown that no steady waves exist for $r \geq R_0$ (under condition (iii)). We note that there is no analogues statement for irrotational waves. A typical example of a vorticity function satisfying condition (ii) is a negative constant vorticity $\omega(p) = - b$, $b>0$. It is known (see \cite{Wahlen09}) that vorticity distributions of this type give rise to Stokes waves over flows with internal stagnation points, that exist for all Bernoulli constants $r > R_0$. Furthermore, a recent study \cite{Kozlov2020} shows that there exist continuous families of such Stokes waves that approach a solitary wave in the long wavelength limit. The latter solitary wave has $r > R_0$ and rides a supercritical unidirectional flow (corresponding to one of stream solutions $U(y;s)$ with $s>s_c$) but has a near-bottom stagnation point on a vertical line passing through the crest. Thus, even so there are no unidirectional waves for $r > R_0$, there exist Stokes and solitary waves with $r > R_0$ violating assumption \eqref{uni}. These considerations show that the statement of Theorem \ref{thm:main} is sharp in a certain sense. On the other hand, inequality $r \geq R_0 - \Omega(1)$ is not sharp and probably can be improved further. However it is not clear if one can omit completely the term $-\Omega(1)$ from the bound on the Bernoulli constant.

Inequality $r \leq R_0$ for solitary waves puts a natural upper bound for the Froude number 
\[
F^2(s) = \left(\int_0^d (U_y(y;s))^{-2} dy \right)^{-1}.
\]
It is well known that for irrotational solitary waves $F < \sqrt{2}$; see \cite{starr}, \cite{kp}. Furthermore, the bound $F<2$ for rotational waves with a negative vorticity was obtained in \cite{Wheeler2015}. For small negative vorticity distributions inequality $1<F(s)<2$ is stronger than $R_c < R(s) < R_0$. However, already for $\omega(p) = -1$ the inequality $R(s) < R_0$ becomes stronger. For $\omega(p) = - b$ with a large $b>0$ we find that inequality $R_c < R(s) < R_0$ is equivalent to $1<F(s)\lesssim \sqrt{2}$, which is significantly better than $F<2$.
\section{Preliminaries}

\subsection{Reformulation of the problem}

Under assumption \eqref{uni} we can apply the partial hodograph transform introduced by Dubreil-Jacotin \cite{DubreilJacotin34}. More precisely, we present new independent variables
\[
q = x, \ \ p = \psi(x,y),
\]
while new unknown function $h(q,p)$ (height function) is defined from the identity
\[
h(q,p) = y.
\]
Note that it is related to the stream function $\psi$ through the formulas
\begin{equation} \label{height:stream}
\psi_x = - \frac{h_q}{h_p}, \ \ \psi_y = \frac{1}{h_p},
\end{equation}
where $h_p > 0$ throughout the fluid domain by \eqref{uni}. An advantage of using new variables is in that instead of two unknown functions $\eta(x)$ and $\psi(x,y)$ with an unknown domain of definition, we have one function $h(q,p)$ defined in a fixed strip $S = \R \times [0,1]$. An equivalent problem for $h(q,p)$ is given by

\begin{subequations}\label{height}
	\begin{alignat}{2}
	\label{height:main}
	\left( \frac{1+h_q^2}{2h_p^2} + \Omega \right)_p - \left(\frac{h_q}{h_p}\right)_q &=0 &\qquad& \text{in } S,\\
	\label{height:top}
	\frac{1+h_q^2}{2h_p^2} +  h  &= r &\quad& \text{on }p=1,\\
	\label{height:bot} 
	h  &= 0 &\quad& \text{on }p=0.
	\end{alignat}
\end{subequations}

The wave profile $\eta$ becomes the boundary value of $h$ on $p = 1$:
\[
h(q,1) = \eta(q), \ \ q \in \R.
\]
Using \eqref{height:stream} and Bernoulli's law \eqref{eqn:bernoulli} we recalculate the flow force constant $\FF$ defined in \eqref{flowforce} as
\begin{equation}\label{height:ff}
\FF = \int_0^1 \left( \frac{1-h_q^2}{h_p^2} - h - \Omega + \Omega(1) + r \right) h_p \, dp.
\end{equation}
Laminar flows defined by stream functions $U(y; s)$ correspond to height functions $h = H(p; s)$ that are independent of horizontal variable $q$. The corresponding equations are
\[
\left(\frac{1}{2H_p^2} + \Omega\right) = 0, \ \ H(0) = 0, \ \ H(1) = d(s), \ \ \frac{1}{2 H_p^2(1)} + H(1) = R(s).
\]
Solving equations for $H(p; s)$ explicitly, we find
\[
H(p;s) = \int_0^p \frac{1}{\sqrt{s^2 -2\Omega(\tau)}} \, d\tau.
\]
Given a height function $h(q,p)$ and a stream solution $H(p;s)$, we define
\begin{equation}\label{ws}
\ws(q,p) = h(q,p) - H(p;s).
\end{equation}
This notation will be frequently used in what follows. In order to derive an equation for $\ws$ we first write \eqref{height:main} in a non-divergence form as
\[
\frac{1+h_q^2}{h_p^2} h_{pp} - 2\frac{h_q}{h_p} h_{qp} + h_{qq} - \omega(p) h_p = 0.
\]
Now using our ansats \eqref{ws}, we find
\begin{equation}\label{ws:main}
\frac{1+h_q^2}{h_p^2} \ws_{pp} - 2\frac{h_q}{h_p} \ws_{qp} + \ws_{qq} - \omega(p) \ws_p + \frac{(\ws_q)^2 H_{pp}}{h_p^2} - \frac{\ws_p (h_p + H_p) H_{pp}}{h_p^2 H_p^2} = 0.
\end{equation}
Thus, $\ws$ solves a homogeneous elliptic equation in $S$ and is subject to a maximum principle; see \cite{Vitolo2007} for an elliptic maximum principle in unbounded domains. The boundary conditions for $\ws$ can be obtained directly from \eqref{height:top} and \eqref{height:bot} by inserting
\eqref{ws} and using the corresponding equations for $H$. This gives
\begin{subequations}\label{ws:boundary}
\begin{alignat}{2}
\frac{(\ws_q)^2}{2h_p^2} - \frac{\ws_p (h_p + H_p)}{2h_p^2 H_p^2} + \ws &=r- R(s) &\qquad& \text{for } p=1,\label{ws:top} \\ 
\ws &= 0&\qquad& \text{for } p=0. \label{ws:bot}
\end{alignat}
\end{subequations}
Concerning the regularity, we will always assume that $\omega \in  C^\gamma
([0; 1])$ and $h \in  C^{2,\gamma}
(\overline{S})$, where $C^{2,\gamma}(\overline{S})$
 is the usual subspace of $C^2(\overline{S})$ (all partial derivatives up to the second order are bounded
and continuous in $\overline{S}$) of functions with H\"older continuous second-order derivatives with a finite H\"older norm, calculated over the whole strip $S$. The exponent 
$\gamma \in (0; 1)$ will be fixed throughout the paper. The Bernoulli constant $r$ will remain unchanged and we will often omit it from notations, such as $s_\pm$ or $\FF_\pm$. Furthermore, in many formulas such as \eqref{ws:top},
we will omit the dependence on $s$ in the notation for $H$, while the right choice of $s$ will be clear from the context.

\subsection{Auxiliary functions $\sigma$ and $\kappa$}

For a given $r > R_c$ and $s > s_0$ we define
\begin{equation} \label{sigma}
	\sigma(s;r) = \int_0^1 \left( \frac{1}{2H_p^2(p;s)} - H(p;s) - \Omega(p) + \Omega(1) + r \right) H_p(p;s) \,dp.
\end{equation}
This expression coincides with the flow force constant for $H(p; s)$, but with the Bernoulli constant $R(s)$ replaced by $r$. We also note that
\[
\sigma(s_\pm(r);r) = \FF_\pm(r).
\]
The key property of $\sigma(s; r)$ is stated below.

\begin{lemma} \label{lemma:sigma} For a given $r \geq R_0$ the function $s \mapsto \sigma(s;r)$ decreases for $s \in (s_0,s_+(r))$ and increases to infinity for $s \in (s_+(r),+\infty)$.
\end{lemma}
\begin{proof}
Because
\[
H_p(p;s) = \frac{1}{\sqrt{s^2 - 2 \Omega(p)}}, \ \ \partial_s H_p(p;s) = -s H_p^3(p;s),
\]
we can compute the derivative
\[
\begin{split}
\sigma_s(s;r) & = \int_0^1 \left( \frac{1}{2H_p^2(p;s)} - H(p;s) - \Omega(p) + \Omega(1) + r \right) \partial_s H_p(p;s) \,dp  \\
& +  \int_0^1 \left( -\frac{\partial_s H_p(p;s)}{H_p^3(p;s)} - \partial_s H(p;s) \right) H_p(p;s) \,dp \\
& = \int_0^1 \left( -\frac{1}{2H_p^2(p;s)} - \Omega(p) + \Omega(1) + r \right) \partial_s H_p(p;s) \,dp - d(s)d'(s) \\
& = \int_0^1 \left( -\tfrac12 s^2 + \Omega(1) + r \right) \partial_s H_p(p;s) \,dp - d(s) \int_0^1 \partial_s H_p(p;s) \\
& = -s (r-R(s)) \int_0^1 H_p^3(p;s) \, dp.
\end{split}
\]
Finally, because $R(s) < r$ for $s_0 < s < s_+(r)$ and $R(s) > r$ for $s> s_+(r)$ we obtain the statement of the lemma.
\end{proof}

Our function $\sigma(s;r)$ and it's role is similar to the function $\sigma(h)$ introduced by Keady and Norbury in \cite{Keady1975}. The main purpose of the latter is to be used for a comparison with the flow force constant $\FF$.  

The following function will be also involved in our analysis.
\begin{equation} \label{kappa}
	\kappa(s;r) = 2 (\FF - \sigma(s;r)) - (r- R(s))^2.
\end{equation}
A direct computation gives
\[
\begin{split}
\partial_s \kappa(s;r) & = - 2 \partial_s \sigma(s;r) + 2 (r - R(s)) R'(s) \\
& = 2 s (r-R(s)) \int_0^1 H_p^3(p;s) \, dp + 2 (r - R(s)) (s + d'(s)) \\
& = 2 s (r - R(s)).
\end{split}
\]
Thus, we obtain
\begin{lemma} \label{lemma:kappa} For a given $r \geq R_0$ the function $s \mapsto \kappa(s;r)$ increases for $s \in (s_0,s_+(r))$ and decreases to minus infinity for $s \in (s_+(r),+\infty)$.
\end{lemma}

Properties of functions $\sigma$ and $\kappa$ will used in what follows.

\subsection{Flow force flux functions}

Our aim is to extract some information by comparing the flow force constant $\FF$ (of a given solution with the Bernoulli constant $r \geq R_0$) to $\sigma(s;r)$ for different values of $s > s_0$. For this purpose we first compute the difference
\[
\begin{split}
\FF - \sigma(s;r) & = \int_0^1 \left( \frac{1-(\ws_q)^2}{2h_p^2} - \ws - \frac{1}{2 H_p^2} \right) H_p \, dp \\
& + \int_0^1 \left( \frac{1-(\ws_q)^2}{2h_p^2} - h - \Omega + \Omega(1) + r \right) \ws_p \, dp \\
& = \int_0^1 \left( \frac{(\ws_p)^2}{2h_p H_p^2} - \frac{(\ws_q)^2}{2h_p} + \ws H_p \right) \, dp \\
& + \int_0^1 \left( -\frac{1}{2H_p^2} - \ws - H - \Omega + \Omega(1) + r  \right) \ws_p \, dp.
\end{split}
\]
Now using the identity
\[
- \Omega(p) + \Omega(1) + R(s) = \frac{1}{2H_p^2} + H(1)
\]
and integrating first-order terms, we conclude that
\[
2(\FF - \sigma(s;r)) = 2 (r - R(s)) \ws(q,1) - (\ws(q,1))^2 + \int_0^1 \left( \frac{(\ws_p)^2}{h_p H_p^2} - \frac{(\ws_q)^2}{h_p} \right) \, dp.
\]
Let us define the (relative) flow force flux function $\Phi^{(s)}$ by setting
\begin{equation} \label{fff}
\phis(q,p) = \int_0^p \left( \frac{(\ws_p(q,p'))^2}{h_p(q,p') (H_p(p';s))^2} - \frac{(\ws_q(q,p'))^2}{h_p(q,p')} \right) \, dp'.
\end{equation}
An analog (partial case with $s = s_+(r)$) of this function was recently introduced in \cite{KozLokhWheeler2020}. The same computation as in \cite{KozLokhWheeler2020} gives
\begin{equation} \label{fff:der}
\phis_q = - \ws_q \left( \frac{1 + (\ws_q)^2}{h_p^2} - \frac{1}{H_p^2}\right), \ \ \phis_p = \frac{(\ws_p)^2}{h_p H_p^2} - \frac{(\ws_q)^2}{h_p}.
\end{equation}
A surprising fact about $\phis$ is that it solves a homogeneous elliptic equation as stated in the next proposition.

\begin{proposition}
	There exist functions $b_1, b_2 \in L^{\infty}(S)$ such that
	\begin{equation} \label{fff:main}
	\frac{1+h_q^2}{h_p^2} \phis_{pp} - 2\frac{h_q}{h_p} \phis_{qp} + \phis_{qq} + b_1 \phis_q + b_2 \phis_p = 0 \ \ \text{in} \ \ S.
	\end{equation}
	Furthermore, $\phis$ satisfies the boundary conditions
\begin{subequations}\label{fff:boundary}
	\begin{alignat}{2}
	\phis &= 2(\FF - \sigma(s;r)) - 2 (r - R(s)) \ws(q,1) + (\ws(q,1))^2 &\qquad& \text{for } p=1,\label{fff:top} \\ 
	\phis &= 0&\qquad& \text{for } p=0. \label{fff:bot}
	\end{alignat}
\end{subequations}
	In the irrotational case $b_1,b_2 = 0$ and \eqref{fff:main} is equivalent to the Laplace equation.
\end{proposition}

For the proof we refer to \cite{KozLokhWheeler2020}. We also
note that $\phis \in C^{2,\gamma}(\overline{S})$, provided $h \in C^{2,\gamma}(\overline{S})$ and $\omega \in C^{\gamma}([0,1])$.

The next proposition explains the meaning of the auxiliary function $\kappa(s; r)$.

\begin{proposition} \label{p:fff}
	Let $h \in  C^{2,\gamma}(\overline{S})$ be a solution to \eqref{height} with $r>R_c$. Assume that the flow force flux function $\phis$ for some $s > s_0$ satisfies $\inf_{q \in \R} \phis(q; 1) \leq 0$. Then
	\[
	\inf_{q \in \R} \phis(q; 1) = \kappa(s; r),
	\]
	where $\kappa(s; r)$ is defined by \eqref{kappa}.
\end{proposition}

\begin{proof}
	First, we assume that the infimum is attained at some point $(q_0; 1)$, where $\phis_q(q_0; 1) = 0$. Differentiating the boundary condition \eqref{fff:top}, we find
	\begin{equation} \label{p:fff:eq2}
	\phis_q(q_0, 1) = 2\ws_q(q_0, 1) (\ws(q_0, 1) - (r-R(s))) = 0.
	\end{equation}	
	Because $\phis$ attains it's global minimum at $(q_0,1)$, then the maximum principle and the Hopf lemma give $\phis_p(q_0,1) < 0$. In particular, we find that $\ws_q(q_0, 1) \neq 0$ by the second formula \eqref{fff:der}. Thus, we necessarily obtain
	\[
	\ws(q_0, 1) = (r-R(s)).
\]
	Using this equality in \eqref{fff:der}, we conclude $\phis(q_0,1) = \kappa(s;r)$ as required.
	
	Now we assume that the infimum is attained over a sequence $\{q_j\}_{j=1}^\infty$ accumulating at the
	positive infinity. Passing to a subsequence, if necessary, we can assume that
	\begin{equation}\label{p:fff:eq1}
	\lim_{j \to +\infty} \phis_q(q_j,1) = 0, \ \ \lim_{j \to +\infty} \phis_p(q_j,1) \leq 0.
	\end{equation}
	There are two possibilities:
	\[
	 (i) \ \lim_{j \to +\infty} \ws_q(q_j,1) = 0 \ \ \text{and} \ \  (ii) \ \lim_{j \to +\infty} \ws_q(q_j,1) \neq 0. 
	\]	 
	 In the first case relations in \eqref{p:fff:eq1} give
	\[
	 \lim_{j\to+\infty} \ws_q(q_j,1) = \lim_{j\to+\infty} \ws_p(q_j,1) = 0,
	\]
	which then require 
	\[
	\lim_{j\to+\infty} \ws(q_j,1) = r - R(s)
	\]
	by the Bernoulli equation \eqref{ws:top}. In this case $\inf_S \phis = \kappa(s;r)$ as desired. The remaining option (ii) provides with a subsequence $\{q_{j_k}\}$ such that $\lim_{k\to+\infty} \ws(q_{j_k},1) = r - R(s)$, which follows from the first relation in \eqref{p:fff:eq1} and \eqref{p:fff:eq2}. Thus, we find again that $\inf_S \phis = \kappa(s;r)$, which completes the proof.
\end{proof}

\section{Proof of Theorem \ref{thm:main}}

Assume that the vorticity function $\omega$ satisfies condition (ii) of the theorem. In this case $d_0, R_0 < +\infty$, $s_0 = 0$ and 
\begin{equation} \label{condii}
\inf_{s > s_0} H_p(0;s) = +\infty.
\end{equation}

First we prove the claim about solitary waves. Thus, we assume that there exists a solitary wave solution $h$ with $r \geq R_0$. Choosing $s = s_+(r)$, we put
\[
w(q,p) = h(q,p) - H(p;s_+(r)).
\]
It follows from Theorem 1 in \cite{Kozlov2015} that $w(q,1) > 0$ for all $q \in \R$. Now because for a supercritical solitary wave $\FF = \sigma(s_+(r);r)$ and the relation \eqref{fff:top} is then reduced to
\[
\Phi^{(s_+(r))} = (\ws)^2,
\]
we find that $\Phi^{(s_+(r))}$ is strictly positive along the top boundary. On the other hand, we can choose $s \in (s_0,s_+(r))$ sufficiently small so that $\ws_p (q_0,0) = 0$ for some $q_0 \in \R$, which follows from \eqref{condii}. Then the corresponding flow force flux function $\phis$ must attain negative values somewhere along the top boundary, because otherwise $\phis_p(q,0) > 0$ for all $q \in \R$ by the Hopf lemma, leading to a contradiction with $\ws_p (q_0,0) = 0$ in view of the second formula \eqref{fff:der}. Since $\phis$ depends smoothly on $s$, by the continuity we can find $s_\star \in (s_0,s_+(r))$ for which $\inf_{q \in \R} \Phi^{(s_\star)}(q,1) = 0$. By Proposition \ref{p:fff} we obtain $\kappa(s_\star;r) = 0$ so that $\FF > \sigma(s_\star;r)$. Now Lemma \ref{lemma:sigma} gives $\sigma(s_\star;r) > \sigma(s_+(r);r)$ and then $\FF > \sigma(s_+(r);r)$, which can not be true for a supercritical  solitary wave.

Now we consider the case of a Stokes wave $h$ for some $r \geq R_0$. Our aim is to show that $r < R_0 - \Omega(1)$. We start by proving 
\begin{lemma} \label{lemma:stokes}
	There exists $s_\star \in (s_0,s_+(r))$ such that $\FF < \sigma(s_\star;r)$.
\end{lemma}
\begin{proof}
	
	Let $q_t < q_c$	be coordinates for some adjacent trough and crest respectively, so that $h(q,1)$ is monotonically increasing on the interval $(q_t,q_c)$. By \eqref{condii} we can choose a stream solution $H(p;s_\star)$ with $s_\star \in (s_0,s_+(r))$ such that $h_p(q_\star,0) = H_p(0;s_\star)$ for some $q_\star \in (q_t,q_c)$. For the function 
	\[
	w^{(\star)}(q,p) = h(q,p) - H(p;s_\star)
	\]
	we consider the zero level set
	\[
	\Gamma = \{ (q,p) \in (q_t,q_c) \times (0,1): \ \ w^{(\star)}(q,p) = 0 \}
	\]
	inside the rectangle $Q = (q_t,q_c) \times (0,1)$. 	We claim that $\Gamma$ is a graph $\{(f(p),p), \ \ p \in (0,1) \}$ of some function $f \in C^{2,\gamma}([0,1])$ such that $f(0) = q_\star$ and $f(1) \in (q_t,q_c)$. Thus, the curve $\Gamma$ connects a point on the bottom with the surface. To explain this fact we need to recall some properties of Stokes waves. Let $Q_l,Q_r,Q_t$ and $Q_b$ be the left, right, top and bottom boundaries of $Q$, excluding corner points. Then the following properties are true:
	\begin{itemize}
		\item[(a)] $w^{(\star)}_q > 0$ on $Q$, while $w^{(\star)}_q = 0$ on $Q_l,Q_r$ and $Q_b$;
		\item[(b)] $w^{(\star)}_{qq} > 0$ on $Q_l$;
		\item[(c)] $w^{(\star)}_{qq} < 0$ on $Q_r$;
		\item[(d)] $w^{(\star)}_{qp} > 0$ on $Q_b$.
	\end{itemize}
	First of all, (a) guarantees that $\Gamma$ (if not empty) is locally the graph of a function as desired. We only need to show that it connects $Q_t$ and $Q_b$. Note that $w^{(\star)}$ attains a unique zero value at some point $(q_\dagger,1)$ on $Q_t$. Otherwise, we would find that $w^{(\star)}_p(q,0)$ has a constant sign by the Hopf lemma, contradicting to the equality $w^{(\star)}_p(q_\star,0) = 0$. Thus, $\Gamma$ bifurcates locally from $(q_\dagger,1)$ inside $Q$. On the other hand, (d) shows that $\Gamma$ also bifurcates inside $Q$ from $(q_\star,0)$ on the bottom. Now it is easy to see that theses two curves must be connected with each other. Indeed, relations (b) and (c) and inequalities $w^{(\star)}_p(q_t,0) < 0 < w^{(\star)}_p(q_c,0)$ guarantee that $w^{(\star)}_p$ has constant sign on the vertical sides $Q_l$ and $Q_r$. In particular, $w^{(\star)}$ is strictly negative on $Q_l$ and positive on $Q_r$. Thus, $\Gamma$ can not approach sides $Q_l$ and $Q_r$ and must connect $Q_t$ and $Q_b$ as desired.
	
	Now we can prove that $\Phi^{(s_\star)}(q_\dagger,1) < 0$ and then $\FF < \sigma(s_\star;r)$ by \eqref{fff:top}, since $w^{(\star)}(q_\dagger,1) = 0$. For that purpose we compute $\Phi^{(s_\star)}(q_\dagger,1)$ by changing a contour of integration as follows:
	\[
	\Phi^{(s_\star)}(q_\dagger,1) = \int_0^1 \Phi^{(s_\star)}_p(q_\dagger,p) \, dp = \int_{\Gamma} (\Phi^{(s_\star)}_p,-\Phi^{(s_\star)}_q) \cdot \boldsymbol{n} \, \textrm{dl},
	\]
	where $\textrm{dl}$ is the length element and $\boldsymbol{n} = (n_1,n_2)$ is the unit normal to $\Gamma$ with $n_1 > 0$ (because $\Gamma$ is the graph of $f(p)$). Note that $\boldsymbol{n}$ is proportional with $(w^{(\star)}_q,w^{(\star)}_p)$ along $\Gamma$ and is oriented in the same way. Therefore, $(\Phi^{(s_\star)}_p,-\Phi^{(s_\star)}_q) \cdot \boldsymbol{n}$ has the same sign as
	\begin{equation} \label{fff:field}
	(\Phi^{(s_\star)}_p,-\Phi^{(s_\star)}_q) \cdot (w^{(\star)}_q,w^{(\star)}_p) = - \left( \frac{(w^{(\star)}_p)^2}{h_p^2 H_p} + \frac{(w^{(\star)}_q)^2 H_p}{h_p^2} \right) w^{(\star)}_q < 0,
	\end{equation}
	which is a matter of a straightforward computation based on \eqref{fff:der}. To see that we first rewrite $\phis_q$ as
	\[
	\phistar_q =  \wstar_q \left( \frac{\phistar_p}{h_p} + \frac{2 \wstar_p}{h_p^2 H_p}  \right).
	\]
	Using this formula we compute
	\[
	\begin{split}
	(\Phi^{(s_\star)}_p,-\Phi^{(s_\star)}_q) \cdot (w^{(\star)}_q,w^{(\star)}_p) & = \phistar_p \wstar_q - \wstar_q \wstar_p \left( \frac{\phistar_p}{h_p} + \frac{2 \wstar_p}{h_p^2 H_p}  \right) \\
	& = \wstar_q \left( \frac{H_p \phistar_p}{h_p} - \frac{2(\wstar_p)^2}{h_p^2 H_p} \right) 
	\end{split}
	\]
	It is left to use formula \eqref{fff:der} for $\phistar_p$ to conclude \eqref{fff:field}. Thus, $(\Phi^{(s_\star)}_p,-\Phi^{(s_\star)}_q) \cdot \boldsymbol{n}$ is negative along $\Gamma$ and then $\Phi^{(s_\star)}(q_\dagger,1) < 0$. The lemma is proved.
\end{proof}

Using Lemma \ref{lemma:stokes} it is easy to complete the proof of the theorem. Indeed, for all $s \in (s_0,s_\star)$ we have $\FF < \sigma(s_\star;r)$, while at the every crest we have $\phis(q_c,1) > 0$, because of \eqref{fff} and that $\ws_q(q_c,p) = 0$ for all $p \in [0,1]$. Thus, the boundary condition \eqref{fff:top} then implies
\[
\ws(q_c,1) > 2 (r-R(s)),
\]
which is true for all $s \in (s_0,s_\star)$. Here we used the fact that $\ws(q_c,1) > 0$, which was proved in \cite{Kozlov2015}. Passing to the limit $s \to s_0$, we find
\[
\eta(q_c) > d_0 + 2 (r - R_0).
\]
Finally, because $\eta(q_c) < r$ by \eqref{height:top} and $R_0 = d_0 - \Omega(1)$, we obtain
\[
r < R_0 - \Omega(1),
\]
which finises the proof of the theorem.

\bibliographystyle{alpha}
\bibliography{bibliography}
\end{document}